\documentclass[a4paper,UKenglish,autoref, unicode,hidelinks, thm-restate]{lipics-v2021}

\nolinenumbers

\usepackage{mathrsfs}
\usepackage{xcolor}

\usepackage{amsmath}
\usepackage{amssymb}
\usepackage{amsthm}
\usepackage{thm-restate}
\usepackage{mdframed}
\usepackage{mathtools}
\usepackage[capitalize]{cleveref}

\usepackage{float}
\usepackage{wrapfig}

\usepackage{cancel}
\usepackage{xspace}
\usepackage[textsize=tiny]{todonotes}
\usepackage[all,defaultlines=3]{nowidow}
\usepackage{microtype}

\newcommand{\N}{\mathbb{N}}
\mdfdefinestyle{myenvs}{%
  hidealllines=true,%
  nobreak=true, 
  leftmargin=0pt,
  rightmargin=0pt,
  innerleftmargin=0pt,
  innerrightmargin=0pt,
}

\usepackage{dirtytalk}
\newcommand{\Cc}[0]{\ensuremath{\mathscr{C}}\xspace}

\newcommand{\CC}[0]{\mathrm{\mathscr{C}}}
\newcommand{\DD}[0]{\mathrm{\mathscr{D}}}
\newcommand{\NN}[0]{\mathrm{\mathbb{N}}}

\newcommand{\C}{\mathcal{C}}
\newcommand{\Dcal}{\ensuremath{\mathcal{D}}}
\newcommand{\Tcal}{\ensuremath{\mathcal{T}}}

\newcommand{\flip}[1]{\mathsf{#1}}
\newcommand{\FF}{\mathcal{F}}

\renewcommand{\phi}{\varphi}
\renewcommand{\le}{\leqslant}
\renewcommand{\ge}{\geqslant}
\renewcommand{\leq}{\le}
\renewcommand{\geq}{\ge}

\newcommand{\ramsey}[2]{\ensuremath{\mathscr{R}(#1,#2)}\xspace}
\newcommand{\fo}{\textsc{FO}\xspace}
\newcommand{\np}{\ensuremath{\mathsf{NP}}\xspace}

\def\ie{\emph{i.e.}\xspace}

\newtheorem*{maintheorem}{Main Theorem}

\title{Weakly-sparse and strongly flip-flat classes of graphs are uniformly almost-wide} 

  \titlerunning{Sparse strongly flip-flats are almost-wide } 

\author{Fatemeh Ghasemi}{LACL, Universit\'{e} Paris-Est Cr\'{e}teil, France}{fatemeh.ghasemi@lacl.fr}{}{}
\author{Julien Grange}{LACL, Universit\'{e} Paris-Est Cr\'{e}teil, France}{julien.grange@lacl.fr}{https://orcid.org/0009-0005-0470-1781}{}
\author{Mamadou Moustapha Kant\'{e}}{Universit\'{e} Clermont Auvergne, Clermont Auvergne INP, LIMOS, CNRS, Aubi\`{e}re, France}{mamadou.kante@uca.fr}{https://orcid.org/0000-0003-1838-7744}{}
\author{Florent Madelaine}{LACL, Universit\'{e} Paris-Est Cr\'{e}teil, France}{florent.madelaine@u-pec.fr}{https://orcid.org/0000-0002-8528-7105}{}

\authorrunning{F. Ghasemi, J. Grange, M.M. Kant\'{e}, F. Madelaine} 

\Copyright{Jane Open Access and Joan R. Public} 

\ccsdesc[500]{Theory of computation~Finite Model Theory}
\ccsdesc[500]{Mathematics of computing~Graph algorithms}
\keywords{Almost-wide, Flip-flatness} 

\category{} 

\relatedversion{} 

\EventEditors{John Q. Open and Joan R. Access}
\EventNoEds{2}
\EventLongTitle{42nd Conference on Very Important Topics (CVIT 2016)}
\EventShortTitle{CVIT 2016}
\EventAcronym{CVIT}
\EventYear{2016}
\EventDate{December 24--27, 2016}
\EventLocation{Little Whinging, United Kingdom}
\EventLogo{}
\SeriesVolume{42}
\ArticleNo{23}

\begin{document}

\maketitle

\begin{abstract}
In this work we take a step towards characterising \emph{strongly flip-flat} classes of graphs. 
Strong flip-flatness appears to be the analogue of uniform almost-wideness in the setting of dense classes of graphs.
We prove that strongly flip-flat classes of graphs that are \emph{weakly sparse} are indeed uniformly almost-wide. 
\end{abstract}

\section{Introduction}
When Trakhtenbrot showed that the set of \emph{first-order} formulas (\fo formulas for short) that are valid on finite structures is \emph{co-recursively enumerable}, but not \emph{recursively enumerable},  a beam of light was shed on the curious inversion of finite structures and infinite structures; as the same set is recursively enumerable (G\"odel's completeness), but not co-recursively enumerable (Church's theorem) on structures. This gave birth to finite model theory \cite{ebbinghaus2006finite}
 and boosted the motivation to show its differences with the world of infinites. 
 Some classical results, such as the \emph{compactness theorem}, almost immediately broke down; some, like the \emph{L\"owenheim-Skolem theorem}, became meaningless; and some remained valid in finite structures, such as \emph{Ehrenfeuchet-Fra\"iss\'e} games. Moreover, many results were discovered to be native of finites and not hold in the infinites such as Fagin showing that on finites structures $\Sigma_1^1= \np$ \cite{baldwin2000finite}.


A category of theorems from classical model theory are \emph{preservation} theorems. They describe the relation between syntactic and semantic properties of first-order formulas. \emph{\L{}o\'s-Tarski, Lyndon's positivity} and \emph{Homomorphism preservation theorem} are examples of preservation theorems \cite{hodges1993model}. The \emph{\L{}o\'s-Tarski} theorem, also known as \emph{extension preservation theorem}, which asserts that a first-order formula is preserved under embeddings between all structures if and only if it is equivalent to an existential formula, was shown to fail in the finite from early on \cite{tait1959counterexample}. On the contrary, Rossman showed that the homomorphism preservation theorem asserting that a formula is preserved by homomorphisms if and only if it is existential-positive, holds on \emph{finite structures} \cite{rossman2008homomorphism}. Before Rossman, Atserias et al. attempted to show positive results regarding both extension and homomorphism preservation  by further restricting finite structures to tame classes defined by sparsity conditions
 (with additional closure conditions)~\cite{atserias2008preservation,atserias2006preservation,dawar2024preservation,dawar2010homomorphism}, and the notions such
 as \emph{wideness, quasi-wideness}, \emph{almost-wideness} and their uniform versions were introduced for that purposes.
 
 These notions proved to be useful in other contexts. Informally, a class of graphs $\mathscr{C}$ is said quasi-wide if, for every integer $r$ and in any large enough graph $G$ 
 from $\mathscr{C}$, we can find a small set of vertices $S$ and a large enough set $B\subseteq V(G)\setminus S$ such that any two vertices in $B$ are at
 distance at least $r+1$ in $G\setminus S$, \ie, the induced graph on $V(G)\setminus S$. Being \emph{uniformly quasi-wide}, means that the above statement is true not only for
 any large enough graph, but also for their large enough subgraphs. The well-known notion of \emph{nowhere dense} classes of graphs are characterised to be
 exactly the uniformly quasi-wide classes of graphs \cite{nevsetvril2012sparsity}, and 
 this characterisation lies at the heart of another combinatorial characterisation of nowhere dense classes of graphs through the notion of \emph{splitter
   game}, which helped proving that \emph{$\fo$ model checking} on nowhere dense classes of graphs is 
 \emph{fixed parameter tractable} \cite{grohe2017deciding}.

 It is therefore natural to ask for which other graph classes such characterisations or algorithmic results can be obtained, and there is a large effort from
 the finite model theory and algorithmic/structural graph theory communities to push them to hereditary dense structures (see for instance
 \cite{dreier2024flip,gajarsky2020first,dreier2024first,torunczyk2023flip} to cite a few). With this aim, the well-studied model-theoretic notion of \emph{monadic stability} beautifully manifested itself.

 The concept of monadic stability was introduced by Baldwin and Shelah \cite{baldwin1985second}, and was mainly studied within \emph{model theory} and in connection with Shelah's \emph{classification theory} \cite{shelah1990classification}. Their work provides an understanding of the structure of monadically stable classes from the model-theoretic perspective.
  However,  describing their combinatorial properties has started only recently. Immediately after, the two combinatorial characterisations of monadically stable graph classes through the notions of \emph{flip-flatness} \cite{dreier2023indiscernibles} and the \emph{flipper game} \cite{flipper-game}, came handy in proving that $\fo$-model checking is fixed parameter tractable on monadically stable classes of graphs \cite{ssmc, dreier2024first}. Informally, flip-flatness means that in a big enough graph we can perform a certain number of \emph{flips}, that is flipping the adjacency between a pair of subsets of vertices, such that in the final structure, there exists a large enough set of vertices that are pairwise far away. When the number of flips is independent from the distance between the far-away subset of vertices and only depends on the class, we call the class \emph{strongly flip-flat}.
  Recently, Eleftheriadis extended the result on extension preservation theorem to a restriction of the strongly flip-flat classes of graphs \cite{eleftheriadis2024extension}.
  
  We recall Eleftheriadis' theorem after explaining the notation. If $G$ is a graph and $F=(A,B)$ is a pair of subsets of vertices of $G$, the \emph{flip of $G$ on $F$}, denoted by $G\oplus F$, is the graph obtained from $G$ by
 flipping the
 edges between vertices in $A$ and vertices in $B$, where $A$ and $B$ can have non-empty intersection. Flipping edges between $A$ and $B$  means replacing every
 edge by a non-edge and every non-edge by an edge. One can check that $G\oplus F\oplus F' = G\oplus F'\oplus F$. A family $\mathcal{F}$ of $k$ pairs of subsets of vertices is called a \emph{$k$-flip in $G$}, and for
 a \emph{$k$-flip} $\mathcal{F} := \{F_1,\ldots,F_k\}$ in $G$, we denote by $G \oplus \mathcal{F}$ the graph $G\oplus F_1 \oplus F_2\oplus \cdots \oplus
 F_k$, and by $G\star_{\mathcal{F}} G$ the graph $G'\oplus \mathcal{F'}$ where $G'$ is the disjoint union of two copies $G_1$
 and $G_2$ of $G$ and $\mathcal{F'}=\{(A_1,B_2)\mid (A,B)\in \mathcal{F}, A_1\ \textrm{(resp. $B_2$)}$ the copy of $A\ \textrm{(resp. $B$)}$ in $G_1\ \textrm{(resp. $G_2$)}\}$.

%
 \begin{theorem}[\cite{eleftheriadis2024extension}] Let $\C$ be a hereditary class of graphs. Suppose that there is some $k \in \N$ such that, for all $r \in
   \N$, there is a function $f_r:\N \to \N$ satisfying that, for every $m \in \N$ and every $G \in \C$ of size at least $f_r(m)$, there is a $k$-partition $P$ of
   $V(G)$, some $k$-flip $\mathcal{F} \subseteq P\times P$, and $A \subseteq V(G)$ such that 
   \begin{enumerate}
        \item $|A|\geq m$;
        \item $A$ is $r$-independent in $G \oplus \mathcal{F}$;
        \item\label{ass:3} $G\star_{\mathcal{F}}G \in \C$.
    \end{enumerate}
    Then extension preservation holds over $\C$.
  \end{theorem}

The first two conditions are imposing strong flip-flatness and the third condition is a closure condition which does not naturally hold.
 While it is evident that strongly flip-flat classes of graphs are a subset of flip-flat graph classes, a clear characterisation is still needed.
 
We initiate a characterisation of strongly flip-flat graph classes by looking at their intersection with \emph{weakly sparse} classes of graphs, \ie,
classes of graphs that exclude some biclique. The reason we believe this is a good starting point is that, for many classes of dense graphs, their intersection
with weakly sparse classes of graphs has been observed to be well-behaved classes of graphs, and as such examples we can cite:
 \begin{itemize}
 	\item Bounded shrubdepth $\cap$ Weakly sparse $=$ Bounded treedepth \cite{de2025transducing} (see also the discussion of \cite{gajarsky2022stable}),
 	\item Bounded clique-width $\cap$ Weakly sparse $=$ Bounded tree-width \cite[Theorem 5.9(4)]{courcelle2000upper},
 	\item Bounded linear clique-width $\cap$ Weakly sparse $=$ Bounded path-width \cite{gurski2000tree}. In \cite{gurski2000tree} Gurski and Wanke only
          state the result for graphs of bounded clique-width, which is an improvement on the bound given in \cite[Theorem 5.9(4)]{courcelle2000upper}. However,
          they indeed prove that one can construct a tree decomposition from a given clique-width expression (equivalently NLC-width expression), while keeping the shape of the expression. Hence, when the expression has the form of a path, \ie, linear clique-width is bounded, the resulting tree will look like a path, \ie, path-width is bounded.
 	\item Monadically stable $\cap$ Weakly sparse $=$ Nowhere dense \cite{dvovrak2018induced,podewski1978stable, adler2014interpreting}
 	\item Monadically dependent $\cap$ Weakly sparse $=$ Nowhere dense \cite{dvovrak2018induced} (see also \cite{nevsetvril2021classes}), and
 	\item Bounded flip-width $\cap$ Weakly sparse $=$ Bounded expansion \cite{torunczyk2023flip}.
 \end{itemize}

 Hence, following this line of research, we prove the following.

\begin{maintheorem}
	A class of graphs $\Cc$ is a weakly sparse and strongly flip-flat if and only if $\CC$ is uniformly almost-wide.

\end{maintheorem} 

$\fo$-transductions are a way of constructing new classes of graphs from a given class $\CC$ by the means of $\fo$ formulas. Basically, we have the power to colour
the vertices of the graphs arbitrarily, but with a constant number of colours, and then change the vertex and edge set as the given $\fo$ formulas
dictates. These formulas can use the colours. When a class $\CC$ has property $\mathcal{P}$, for any $\fo$-transduction of $\CC$, we say that it is \emph{structurally}
$\mathcal{P}$. For example, it has been shown that structurally nowhere dense classes of graphs are monadically stable and hence flip-flat \cite{podewski1978stable,
  adler2014interpreting, dreier2023indiscernibles}. On the other hand, it is shown that, for some weakly sparse classes of graphs, their $\fo$-transductions are exactly some well-known, yet
dense classes of graphs. As a prominent example, it is proven that structurally bounded tree-depth classes of graphs are exactly bounded shrubdepth classes of
graphs \cite{ganian2012trees}. In addition to this,  Eleftheriadis also showed that in the setting of hereditary classes of graphs, transductions of uniformly almost-wide classes of graphs are strongly flip-flat
\cite{eleftheriadis2024extension}. Therefore, we conjecture the following, supported by the \emph{sparsification conjecture} stating that monadically stable graph classes are exactly the structurally nowhere dense classes of
graphs.

 \begin{conjecture}\label{conjecture}
 	Every strongly flip-flat graph class is structurally uniform  almost-wide.
 	\end{conjecture}

      We believe  our main theorem is a step towards proving the conjecture, and allow us to confirm it on structurally bounded expansion graph classes. 
Ne{\v{s}}et{\v{r}}il and Ossona De Mendez showed that for a hereditary class of graphs $\CC$, if there are $s\in \NN$ and a function
$t:\NN\rightarrow\NN$ such that, for every $r\in \NN$, the graphs in the depth-$r$ shallow minor of $\CC$ exclude $K_{s,t(r)}$, then $\CC$ is uniformly almost-wide
\cite{nevsetvril2012sparsity}. As second point supporting the conjecture, we also argue that there exist infinitely many $s\in\NN$ such that, for every $r\in\NN$, the
$r$-subdivisions of $K_{s,N}$, for arbitrary large $N$, cannot be strongly flip-flat and we leave open the question of using them to characterise
strongly flip-flat graph classes using \emph{shallow vertex-minors} \cite{buffiere2024shallow}. 

After completing our work, we became aware of independent results by M\"ahlmann \cite[Lemma 13.10]{maehlmann2024}, which establish a statement closely related
to ours. In that work, M\"ahlmann proves that weakly sparse and flip-flat classes of graphs are uniformly quasi-wide. Since
flip-flatness characterizes monadic stability, and uniform quasi-wideness characterizes nowhere denseness, this result mirrors a
corollary of \cite{dvovrak2018induced,nevsetvril2020linear}, which shows that monadically stable and weakly sparse classes of graphs are
nowhere dense. 
While our main results differ, ours concern strong flip-flatness, whereas M\"ahlmann focuses on flip-flatness, the proofs are similar in spirit. Our approach, however, introduces a novel point of view: an induction on the set of flips instead of the usual induction on the independence radius. We believe this perspective brings new insights and contributes meaningfully to the ongoing research in this area.
%


\section{Preliminaries}

For two sets $A$ and $B$, we denote by $A\setminus B$ the set $\{x\in A\mid x\not \in B\}$ and by $A\Delta B$ the set $(A\setminus B)\cup (B\setminus A)$. The
set of positive integers is denoted by $\mathbb{N}$.

We assume familiarity with standard graph theoretical notions, and refer to \cite{bondy2008graph,diestel2024graph}, and deal only with undirected
graphs\footnote{One can extend the results to directed graphs by using a coding of directed graphs, through an invertible \fo-transduction, into bipartite
  undirected graphs.}.  The vertex set of a graph $G$ is denoted by
$V(G)$ and its edge set by $E(G)$. An edge between $u$ and $v$ in a graph $G$ is denoted by $uv$ (equivalently $vu$). 
We recall that, in a graph $G$, a path between two vertices $u$ and $v$ in $V(G)$ is a sequence of vertices $u=w_1,\cdots,w_l=v$ such that $w_i\neq w_j$, for
all $1\leq i< j\leq l$, and, for all $1\leq i< l$, we have that $w_iw_{i+1}\in E(G)$; we say that the length of this path is $l-1$; we often denote such a path
by $v-u$ when we are not interested in the internal vertices of the path. The distance between a pair of vertices is the length of the shortest path between them. For a subset $S$ of $V(G)$, we denote by $G[S]$ the subgraph of $G$ induced by $S$,
\ie, $V(G[S])=S$ and $E(G[S])=(S\times S)\cap E(G)$; by $G\setminus S$ we mean the subgraph of $G$ induced by $V(G)\setminus S$. We say that a set $B\subseteq
V(G)$ is \emph{$r$-independent in $G$} if, for all pairs $u,v\in B$, the distance between $u$ and $v$ in $G$ is strictly greater than $r$. By $K_t$ and $K_{t,t}$, we denote the complete graph on $t$ vertices and the complete bipartite graph with $t$ vertices in each part, respectively. A \emph{half-graph} of order $n$ is a bipartite graph with sides $\{a_1,\cdots,a_n\}$ and $\{b_1,\cdots,b_n\}$ such that, for all $1\leq i,j\leq n$, $a_i$ and $b_j$ are adjacent if and only if $i\leq j$.
\emph{Half-graphs} are graph theoretical representations of linear orders.
%
For an integer $r$, the $r$-subdivision of a graph $G$ is the graph obtained from $G$
by replacing each edge $uv$ of $G$ by a path $u-v$ of length at most $r+1$ that intersects $V(G)$ only on its endpoints. 

\paragraph*{Sparse graph classes.} We adopt the notion of graph sparsity as developed in \cite{nevsetvril2012sparsity}.
\begin{definition}
	Let $\CC$ be a class of graphs. We say that $\CC$ is \textbf{weakly sparse} if there exists a constant $t\in \NN$ such that every graph $G\in \CC$ excludes $K_{t,t}$ as a subgraph. 
\end{definition}

A graph $H$ is a \emph{depth-$r$ shallow minor} of $G$ if there are disjoint connected subgraphs $C_1,\ldots, C_{|V(H)|}$ of
   $G$, each of radius at most $r$, and a bijective mapping $\rho:V(H)\to \{C_1,\ldots,C_{|V(H)|}\}$ such that $uv\in E(H)$ only if there is an edge between a vertex in $\rho(u)$ and a
   vertex in $\rho(v)$.

 A class of graphs $\CC$ is nowhere dense if and only, for every integer $r$, there exists an integer $t_r$ such that all the graphs of $\CC$ exclude $K_{t_r}$ as depth-$r$ shallow minor.
By taking $t$ pairs of vertices from opposite sides of $K_{t,t}$ as disjoint connected subgraphs, one constructs $K_t$ as a depth-$1$ shallow minor.
Therefore, we can conclude that nowhere dense classes of graphs are weakly sparse.
Let us now introduce the core notions of quasi-wideness.

\begin{definition}\label{definition: uniformly quasi-wide and almost-wide classes}
Let $\CC$ be a class of graphs. We say that $\CC$ is \textbf{uniformly quasi-wide} if, for every $r\in \NN$, there exists a function $N_{r}:\NN\rightarrow\NN$
and an integer $s_r$ such that, for all $m\in \NN$ and $G\in \CC$, and for all $A\subseteq V(G)$ of size at least $N_{r}(m)$, we can find sets $S\subseteq V(G)$ and  $B\subseteq A$ where $|S|\leq s_r$ and $|B|> m$, such that $B$ is $r$-independent in $G\setminus S$. A graph class $\CC$ is said to be \textbf{uniformly almost-wide} if $s_r:= s$ is independent of $r$. (See~\autoref{fig:definition of almost wide}.)

We say that $(A,S)$ is an \emph{$(r,m)$-widenable instance of $G$ with respect to $s$ and $N_r$} if $A\subseteq V(G)$ with $|A|\geq N_r(m)$ and $S\subseteq
V(G)$ with $|S|\leq s$, and we can find a set $B\subseteq A$ such that  $|B|\geq m$ and $B$ is $r$-independent in $G\setminus S$. When $s$ and $N_r$ are clear
from the context we only say that $(A,S)$ is an $(r,m)$-widenable instance of $G$.
	
\end{definition}

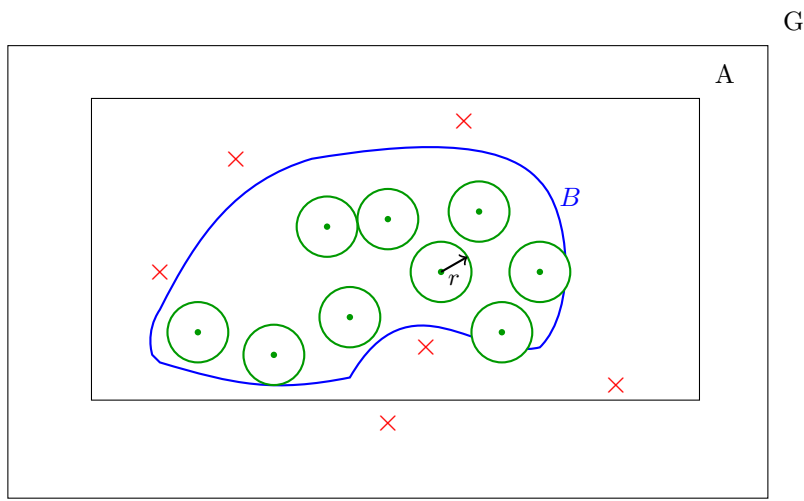
\begin{figure}[h]
\label{fig: fig1}
\center
    \begin{tikzpicture}
  \node[draw, minimum width=10cm, minimum height=6cm] (rect) at (0,0) {};
  \node[draw, minimum width=8cm, minimum height=4cm] (rect2) at (0.1,0.3) {};

  \node[above right=2pt and 2pt of rect] {G};
  \node[above right=2pt and 2pt of rect2] {A};

  \node at (-2,1.5) {\textcolor{red}{\Large$\times$}};
  \node at (1,2) {\textcolor{red}{\Large$\times$}};
  \node at (0,-2) {\textcolor{red}{\Large$\times$}};
  \node at (-3,0) {\textcolor{red}{\Large$\times$}};
  \node at (3,-1.5) {\textcolor{red}{\Large$\times$}};
  \node at (0.5,-1) {\textcolor{red}{\Large$\times$}};

 \draw[thick, blue]
    (-3,-1.2)
    .. controls (-2,-1.5) and (-1.5,-1.6) ..
    (-0.5,-1.4)
    .. controls (0.3,0) and (1.2,-1.2) ..
    (2,-1)
    .. controls (2.5,-0.5) and (2.4,0.8) ..
    (2,1.2)
    .. controls (1.5,1.8) and (0.2,1.7) ..
    (-1,1.5)
    .. controls (-2,1.2) and (-2.5,0.5) ..
    (-3,-0.5)
    .. controls (-3.2,-0.8) and (-3.1,-1.1) ..
    (-3.1,-1.1)
    -- cycle;
  \node[blue] at (2.4,1) {$B$};
\foreach \x/\y in {-2.5/-0.8, -1.5/-1.1, -0.5/-0.6, 0.7/0, 1.5/-0.8,
                     2/0, 1.2/0.8, 0/0.7, -0.8/0.6} {
    \filldraw[draw=green!60!black, fill=white, thick] (\x,\y) circle (0.4);
    \fill[green!60!black] (\x,\y) circle (1.2pt);
  }

  \draw[->, thick] (0.7,0) -- ++(30:0.4) node[midway, below ] {\small $r$};

\end{tikzpicture}
   
    \caption{ Assuming $|A|> N_r(m)$ is a subset of vertices of $G$, a graph in a uniformly quasi-wide class of graphs $\mathscr{C}$, then there are sets of vertices $S$ (the red crosses), with $|S|\leq s_r$,
      and $B\subseteq A$, with $|B|\geq m$, such that $B$ is $r$-independent in  $G\setminus S$. The function $N_r$ and the constant $s_r$ both depend on $r$ and the class
      $\mathscr{C}$. When $s_r$ is independent from $r$ and only depends on the
      class, we call that class uniformly almost-wide.}
       \label{fig:definition of almost wide}
    
\end{figure}

It has been shown that when it comes to hereditary classes of graphs, \ie, closed under induced subgraphs, uniform quasi-wideness characterises
nowhere denseness, and uniform almost-wideness characterises classes for which there exists a function $t:\NN\rightarrow \NN$ and an integer $s$
such that for every $r\in \N$, it excludes $K_{s,t(r)}$ as depth-$r$ shallow minor \cite{nevsetvril2012sparsity}. 

\paragraph*{Beyond sparse graph classes.} While nowhere dense classes of graphs are exactly the graph classes closed under subgraph and for which \fo
model-checking is fixed parameter tractable, there are other important graph classes for which \fo model-checking is still fixed parameter tractable, such as
unit-interval graphs or map graphs. It is conjectured in \cite{ALS2016} that hereditary graph classes for which \fo
model-checking is fixed parameter tractable correspond exactly to monadically dependent graph classes, a model-theoretic notion which will be defined later. The notions introduced in this quest are the ones behind the notion of
strongly flip-flat graph classes.

\begin{definition}[Flips]
Let $G$ be a graph. A \emph{flip} $\flip{F}$ is an operation that is specified by a pair of sets $\flip{F}= (A,B)$, where $A,B\subseteq V(G)$($A$ and $B$ can intersect and even be equal). We denote this operation by $G\oplus\flip{F}$ and it operates on $G$ by complementing the adjacency between $A$ and $B$, more precisely
\begin{align*}
	G\oplus \flip{F} := (V(G),E(G)\Delta ((A\times B)\cup (B\times A)) ).
\end{align*}
By $G\oplus\FF$, where $\FF =\{\flip{F_1},\cdots,\flip{F_k}\}$ is a set of flips, we mean $G\oplus \flip{F_1}\oplus \cdots\oplus\flip{F_k}$. For a flip $\flip{F}=(A,B)$ and a set of vertices $P\subseteq V(G)$, by $\flip{F}\setminus{P}$ we mean $(A\setminus P,B\setminus P)$. To simplify the notation, instead of $(G\setminus P)\oplus \flip{F}\setminus{P}$ we write $(G\setminus P)\oplus \flip{F}$.

%

 \end{definition}
 \begin{observation}\label{Observation: repetition and order of flips}
For any graph $G$ and \emph{flips} $\flip{F},\flip{F'}$ we have $G\oplus \flip{F}\oplus\flip{F'} = G\oplus\flip{F'}\oplus \flip{F}$ and $G\oplus\flip{F}\oplus\flip{F}=G$.
 \end{observation}

\begin{lemma}\label{lemma: flips do not intersect}
	Let $G$ be a graph and $\FF=\{\flip{F_1},\cdots,\flip{F_k}\}$ be a set of flips. There exists a set of flips $\FF'$ where $G\oplus\FF = G\oplus \FF'$ and $|\FF'|\leq \binom{2^{2k}+1}{2}$ such that $\FF'$ does not flip the adjacency of any fixed pair of vertices more than once.
\end{lemma}
\begin{proof}
	Let $\flip{F_1},\cdots,\flip{F_k}$ be equal to $(P_1,P_2),\cdots,(P_{2k-1},P_{2k})$, respectively. 	
Let $\mathsf{Rest}= \{V(G) \setminus \cup_{i \in [2k]} P_i\} $ and 
$$\mathsf{Partition}=\{\cap_{i \in I} P_i \setminus \cup_{j \notin I} P_j \mid  I\subseteq [2k],I\neq \emptyset\} \cup \{\mathsf{Rest}\}.$$

%
	Indeed $\mathsf{Partition}$ is a partition of the vertices of $G$, and for a part $Q\in \mathsf{Partition}$, if $P_i\cap Q\neq\emptyset$, then $Q\subseteq P_i$. Notice that $|\mathsf{Partition}|\leq 2^{2k}+1$.

	This implies that each $P_i$ can be written as a union of some elements of the $\mathsf{Partition}$. Now, assume that for each $P_i$ there is a set of indices $I_i$ such that $P_i=\bigcup_{j\in I_i}Q_j$, and $Q_j$, $j\in I_i$, being elements of $\mathsf{Partition}$, then a flip $F_t=(P_{2t-1}, P_{2t})$ is equivalent to $\FF_t=\{(Q_j,Q_s)\mid j\in I_{2t-1}, s\in I_{2t}\}$. 
	We construct $\FF'$ to be a refinement of $\FF$ by
	taking flips from $\FF_t$, for $1\leq t\leq k$, that appear in an odd number of $\FF_t$.
	 As an even repetition of a flip does not alter the graph (Observation \ref{Observation: repetition and order of flips}), the other flips of $\FF_t$, $1\leq t\leq k$, i.e. those that appear an even number of times, can be avoided. Thus, $\FF'$ has at most $\binom{2^{2k}+1}{2}$ many elements.
	 Now, as the sets in the flips of $\FF'$ are parts of a partition, a fixed pair of vertices occurs in a unique flip. Thus, $\FF'$ does not change the adjacency of a pair of vertices more than once.	
\end{proof}
\begin{corollary}
	A set of flips $\FF=\{\flip{F_1},\cdots,\flip{F_k}\}$ on a graph $G$, can be seen as a binary relation on a partition of the vertex set of $G$. This binary relation specifies which pairs of parts undergo a flip and which do not.
\end{corollary}

We proceed to the formal definition of flip-flatness and strong flip-flatness, these notions are a generalisation of uniformly
  quasi-wideness and uniformly almost-wideness, where \emph{vertex deletion} is replaced by \emph{performing flips}.

\begin{definition}
Let $\CC$ be a class of graphs. We say that $\CC$ is \textbf{flip-flat} if, for every $r\in \NN$, there exists a function $N_r:\NN\rightarrow \NN$ and an integer $s_r$ such that for all $m\in \NN$ and $G\in \CC$, and for all $A\subseteq V(G)$ of size at least $N_r(m)$, we can find a set of flips $\FF$ and $B\subseteq A$ where $|\FF|\leq s_r$ and $|B|\geq m$, such that $B$ is $r$-independent in $G\oplus \FF$. A class $\CC$ is said to be \textbf{strongly flip-flat} if $s_r:= s$ is independent of $r$. (See~\autoref{fig: fig2}.)

We say that $(A,\FF)$ is an \emph{$(r,m)$-flippable instance of $G$} with respect to $N_r$ if $A\subseteq V(G)$ with $|A|\geq N_r(m)$ and $\FF$ is a set of
flips and we can find a set $B\subseteq A$ with $|B|\geq m$ such that $B$ is $r$-independent in $G\oplus \FF$. When $N_r$ is clear from the context, we only say
that $(A,\FF)$ is an $(r,m)$-flippable instance of $G$.
\end{definition}

\begin{figure}[h]

\center
    \begin{tikzpicture}
  \node[draw, minimum width=10cm, minimum height=6cm] (rect) at (0,0) {};
  \node[draw, minimum width=8cm, minimum height=4cm] (rect2) at (0.2,0.5) {};

  \node[above right=2pt and 2pt of rect] {G};
  \node[above right=2pt and 2pt of rect2] {A};

\draw[thick, blue]
    (-3,-1.2)
    .. controls (-2,-1.5) and (-1.5,-1.6) ..
    (-0.5,-1.4)
    .. controls (0.3,0) and (1.2,-1.2) ..
    (2,-1)
    .. controls (2.5,-0.5) and (2.4,0.8) ..
    (2,1.2)
    .. controls (1.5,1.8) and (0.2,1.7) ..
    (-1,1.5)
    .. controls (-2,1.2) and (-2.5,0.5) ..
    (-3,-0.5)
    .. controls (-3.2,-0.8) and (-3.1,-1.1) ..
    (-3.1,-1.1)
    -- cycle;
  \node[blue] at (2.4,1) {$B$};

%
%
%
%
%
%

%
\foreach \x/\y in {-2.5/-0.8, -1.5/-1.1, -0.5/-0.6, 0.7/0, 1.5/-0.8,
                     2/0, 1.2/0.8, 0/0.7, -0.8/0.6} {
    \filldraw[draw=green!60!black, fill=white, thick] (\x,\y) circle (0.4);
    \fill[green!60!black] (\x,\y) circle (1.2pt);
  }

  \draw[->, thick] (0.7,0) -- ++(30:0.4) node[midway, above right] {\small $r$};

\draw[red, thick, rotate=30] (-1.5,2) ellipse (0.3 and 0.6);
\draw[red, thick, rotate=150] (2,0) ellipse (0.3 and 0.6);

\draw[red, thick] (1,2) ellipse (0.3 and 0.6);
\draw[red, thick] (0.3,2) ellipse (0.3 and 0.6);

\draw[red, thick] (-3,0) ellipse (0.3 and 0.6);
\draw[red, thick] (-3.6,0) ellipse (0.3 and 0.6);

\draw[red, thick, rotate=45] (-0.5,-1.5) ellipse (0.3 and 0.6);
\draw[red, thick] (0.3,-1.5) ellipse (0.3 and 0.6);

\draw[red, thick] (3,-1.5) ellipse (0.3 and 0.6);
\draw[red, thick] (3.3,-1.5) ellipse (0.3 and 0.6);

\end{tikzpicture}
\caption{
Assuming $|A|> N_r(m)$ is a subset of vertices of $G$, a graph in a flip-flat class of graphs $\mathscr{C}$, then there is a set of flips $\mathcal{F}$ (the red ellipses), with $|\mathcal{F}|\leq s_r$,
      and a set $B\subseteq A$, with $|B|\geq m$, such that $B$ is $r$-independent in  $G\oplus \mathcal{F}$. The function $N_r$ and the constant $s_r$ both depend on $r$ and the class
      $\mathscr{C}$. When $s_r$ is independent from $r$ and only depends on the
      class, we call that class strongly flip-flat.
  }
  \label{fig: fig2}
   
\end{figure}
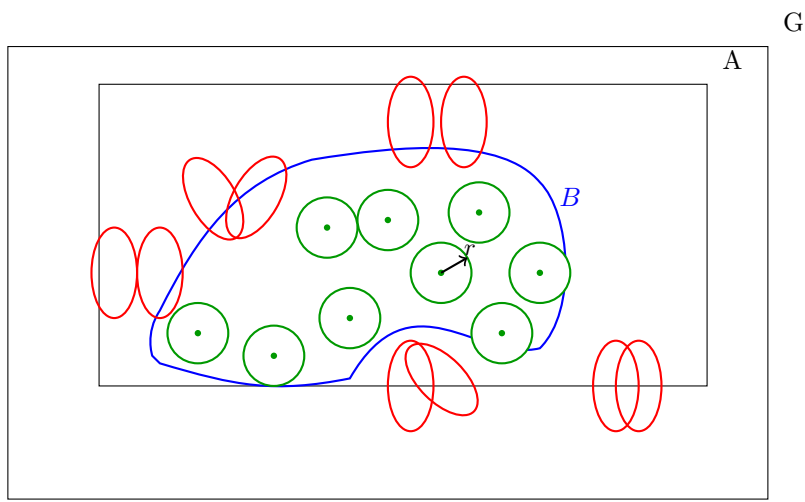

Note that one can mimic the removal of a vertex $v$ by the flip $(\{v\}, N(v))$. 
However, flip-flatness gives us a purely combinatorial characterisation of monadic stability as well, that we recall now. 

One can think of graphs as finite relational structures with the signature $\tau_{E}=\{E\}$, where $E$ is the relation symbol indicating the edge relation. We
refer to \cite{Libkin2004elements} for the definition of \fo formulas. We recall here the notion of a transduction.  Let $c \in \N$ and $\delta(x)$ and $\phi(x,y)$ be
first-order formulas over the extended signature $\tau_E^c:= \tau_E \cup \{C_1,\dots,C_c\}$, where $C_1,\cdots, C_c$ are unary predicates. We think of these
predicates as colours marking the vertices.  The \emph{transduction $T_{\delta,\phi}$} is the operation that maps a graph $G$ to the class $T_{\delta,\phi}(G)$
containing the graphs $H$ such that there exists a $\tau^c_E$-expansion $G^+$ of $G$, i.e., a colouring of $G$ with colours $\{C_1,\dots,C_c\}$, satisfying $V(H)=\{v \in V(G): G^+ \models \delta(v)\}$ and
\[E(H):= \left\{(u,v)\in V(H)^2 \colon u \neq v \land G^+ \models (\phi(u,v) \lor \phi(v,u))\right\}.\]
We write $T_{\delta,\phi}(\C):=\bigcup_{G \in \C} T_{\delta,\phi}(G)$ and say that \emph{$\Dcal$ is a transduction of a class $\C$}, or \emph{$\C$ transduces $\Dcal$}, if there exists a transduction $\Tcal_{\delta,\phi}$ such that  $\Dcal \subseteq \Tcal_{\delta,\phi}(\C)$. A graph class $\C$ is monadically dependent (resp. monadically stable) if $\C$ does not transduce the class of all graphs (resp. half-graphs).  We remark that these are not the original definitions of monadic stability and monadic dependence, but it follows from the work of Baldwin and Shelah that these notions can be defined by transductions \cite{baldwin1985second}.

We say that a class $\CC$ is structurally $\mathcal{P}$, if there exists a transduction $T_{\delta, \varphi}$ and a class $\DD$ with property $\mathcal{P}$ such that $\CC\subseteq T_{\delta, \varphi}(\DD)$.

\begin{theorem}[\cite{dreier2023indiscernibles}]
	A class of graphs is monadically stable if and only if it is flip-flat.
      \end{theorem}

Since an \fo-transduction preserves monadic stability, one can first ask whether \fo-transductions preserve strong flip-flatness, and indeed yes as shown in
\cite{eleftheriadis2024extension}.

\begin{proposition}[\cite{eleftheriadis2024extension}]\label{prop:transduce-sflipflat} If $\mathscr{C}$ is strongly flip-flat, then $T_{\delta, \varphi}(\mathscr{C})$ is also strongly flip-flat,
  for any \fo-transduction $T_{\delta, \varphi}$.
\end{proposition}
      
It has been proven that classes of graphs that are monadically stable and weakly sparse are nowhere dense
\cite{dvovrak2018induced,podewski1978stable, adler2014interpreting}. Thanks to Proposition \ref{prop:transduce-sflipflat}, every \fo-transduction of a uniformly almost-wide graph class is indeed
strongly flip-flat. Hence our motivation to show that classes of graphs that are strongly flip-flat and weakly sparse are uniformly
  almost-wide.

\section{Main result}
We are ready to state and prove the main result.
\begin{theorem}\label{theorem: main result}
	A class of graphs $\Cc$ is a weakly sparse and strongly flip-flat if and only if $\CC$ is uniformly almost-wide.
      \end{theorem}
      
We start with the following lemma.
\begin{lemma}\label{case of k=1}

Let $G$ be a graph, $t_0,r$ and $m$ be integers, and $B\subseteq V(G)$ be a set with $|B|\geq t^2_0+t_0+m-2$, such that all the elements of $B$ are pairwise at distance less than or equal to $r$. Assume there exists a flip $\flip{F}=(P_1,P_2)$ such that:
\begin{itemize}
	\item $\flip{F}$ does not contain a $K_{t_0,t_0}$ with one side in $P_1$ and the other side in $P_2$,
	\item and $B$ is $r$-independent in $G\oplus \flip{F}$.
\end{itemize}
 Then we can find a set $S\subseteq V(G)$, with $|S|\leq t^2_0+t_0-2$ and a set $B'\subseteq B$ with $|B'|\geq m$ such that $B'$ is $r$-independent in
 $G\setminus S$. 
    \end{lemma}

\begin{proof}

	It is safe to assume that $t^2_0+t_0-2\geq 8t_0$, since if we exclude $K_{t_0,t_0}$ then we exclude $K_{t,t}$ for all $t\geq t_0$ as well.
	As all the elements of $B$ are pairwise at distance less than or equal to $r$, we infer that $\flip{F}$ is responsible for breaking the short paths between all pairs
        of vertices from $B$. 
        
        For a path $u=w_1,\cdots,w_l=v$ of length less than or equal to $r$ between two vertices $u$ and $v$, that is broken by $\flip{F}$, we define
        \emph{elementary segments} to be the parts $w_1,w_2,\cdots,w_i$ and $w_j,w_{j+1},\cdots,w_l$ such that $\{w_1,w_2,\cdots,w_i\}\cap(P_1 \cup P_2)= \{w_i\} $ and $\{w_j,w_{j+1},\cdots,w_l\}\cap (P_1 \cup P_2)=\{w_j\}$, i.e., the parts consisting of $v$ or $u$ and only one element of $P_1 \cup P_2$.

        Hence, each such path has
        exactly $2$ elementary segments. As the length of these paths is less than or equal to $r$, and at least one of their edges is between $P_1$ and $P_2$,
        we conclude that at least one of their elementary segments has length strictly less than $\lceil\frac{r}{2}\rceil$; we call this a \emph{short
          elementary segment}. Note that two short elementary segments with different endpoints in $B$ cannot share an endpoint in $P_1\cup P_2$ as this would
        be a path of length less than or equal to $r$ between the elements of $B$ that is not affected by the flip. 
	
	Let us fix $8t_0$ vertices of $B$, say $\{v_1,v_2,\cdots,v_{4t_0}, u_1,\cdots,u_{4t_0}\}$, and focus on the paths in $G$ of length less than or equal to $r$ between pairs
        $(v_1,u_1),(v_2,u_2),\cdots,(v_{4t_0},u_{4t_0})$. These are $4t_0$ paths of length less than or equal to $r$ with different endpoints, having at least $4t_0$ elementary segments that
        have length strictly less than $\lceil\frac{r}{2}\rceil$. At least $2t_0$ of these short elementary segments have an endpoint in the same side of $\flip{F}$. Say, without loss of generality, that $v_1 - a_1, v_2 - a_2, \cdots, v_{2t_0} - a_{2t_0}$ are the short elementary segments, where $a_1, a_2, \cdots, a_{2t_0} \in P_1$, and possibly $v_i = a_i$ for some $i$. Now, we divide the argument based on the parity of $r$.

	\textbf{The even case: $r=2r'$.} Pick $t_0+1$ of these short elementary segments. Say, without loss of generality  $v_1-a_1, v_2-a_2, \cdots,
        v_{t_0+1}-a_{t_0+1}$. Observe that every element in $P_2$ is adjacent to at least $t_0$ of $a_1,a_2,\cdots,a_{t_0+1}$. Indeed, if $u\in P_2$ is not adjacent
        to $a_i,a_j$ where $a_i\neq a_j$, then after the flip we will have that $v_i-a_iua_j-v_j$ is a path of length less than or equal to $r$ between $v_i$ and $v_j$;
        recall that $v_i$ and $v_j$ belong to $B$ and after performing $\flip{F}$ their distance in $G\oplus\flip{F}$ is strictly greater than $r$. For a vertex $a$ among $a_1,a_2,\cdots,a_{t_0+1}$, there are at most $t_0-1$ vertices in $P_2$ that are not adjacent to $a$; as otherwise, these $t_0$ vertices of $P_2$ and $\{a_1,a_2,\cdots,a_{t_0+1}\}\setminus\{a\}$ form a $K_{t_0,t_0}$. Hence, besides these $(t_0-1)(t_0+1)$ vertices of $P_2$, all the remaining form a biclique with $a_1,a_2,\cdots,a_{t_0+1}$. Thus, $|P_2|\leq (t_0-1)(t_0+1)+(t_0-1)$ (Figure \ref{Figure: figure}). We conclude that at least one of $P_1$ or $P_2$ has at most $t_0^2+t_0-2$ many vertices. Without loss of generality, let $P_2$ be the one such that $|P_2|\leq t_0^2+t_0-2$; now, $B\setminus P_2$ is $r$-independent in $G\setminus P_2$, and we have that $|B\setminus P_2|\geq m$.

	\textbf{The odd case: $r=2r'+1$.} Suppose first that at least $t_0$ many of these short elementary segments have length exactly $r'$. Say,
        without loss of generality $v_1-a_1, v_2-a_2, \cdots, v_{t_0}-a_{t_0}$ have length exactly $r'$. Let $b_1-u_1,b_2-u_2,\cdots, b_{t_0}-u_{t_0}$ be the other elementary
        segments of these paths, i.e. $b_1,b_2,\cdots,b_{t_0}\in P_2$ and for all $i=1,\cdots,t_0$ we have that $v_i-a_i$ and $b_i-u_i$ are the elementary
        segments of $v_i-u_i$. As for all $i\in [t_0]$ the length of $v_i-a_i$ is exactly $r'$ we conclude that the length of $b_i-u_i$ is less than or equal to $r'$ as
        well. Therefore, the sets $\{a_1,\cdots,a_{t_0}\}$ and $\{b_1,\cdots,b_{t_0}\}$ form a biclique in $G$. Indeed, if $a_i$ and $b_j$, when $i\neq j$, are not
        adjacent in $G$, then after the flip $v_i-a_ib_j-u_j$ will be a path of length less than or equal to $2r'+1 = r$, which is contradictory and a biclique on $\{a_1,\cdots,a_{t_0}\}$ and $\{b_1,\cdots,b_{t_0}\}$ is also a contradiction to weakly sparseness. Hence, at least $t_0+1$ of the short elementary segments have length strictly less than $r'$. Say, without loss of generality  $v_1-a_1, v_2-a_2, \cdots, v_{t_0+1}-a_{t_0+1}$.
	Now, as in the even case, all the vertices $u\in P_2$ are adjacent to at least $t_0$ of $a_1,a_2,\cdots,a_{t_0+1}$ and we get the same bound on the size of $P_2$. Thus, in this case, $B\setminus P_2$ is also $r$-independent in $G\setminus P_2$, and we have $|B\setminus P_2|\geq m$; hence the statement follows.

	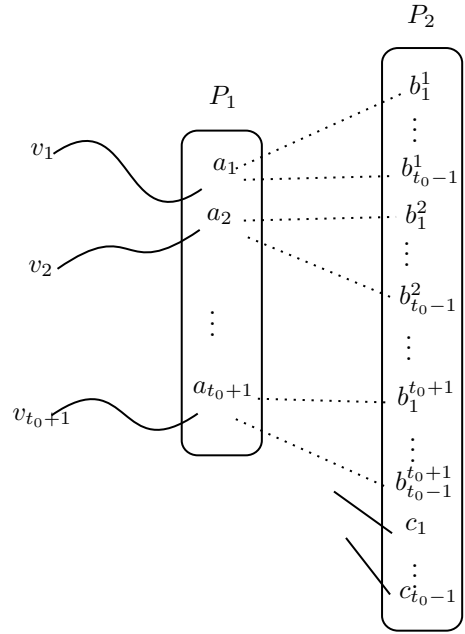
\begin{figure}
	
	\center
			\tikzset{every picture/.style={line width=0.75pt}} 

\begin{tikzpicture}[x=0.75pt,y=0.70pt,yscale=-1,xscale=1]

\draw   (402,107.5) .. controls (406.42,107.5) and (410,111.08) .. (410,115.5) -- (410,413.5) .. controls (410,417.92) and (406.42,421.5) .. (402,421.5) -- (378,421.5) .. controls (373.58,421.5) and (370,417.92) .. (370,413.5) -- (370,115.5) .. controls (370,111.08) and (373.58,107.5) .. (378,107.5) -- cycle ;
\draw   (302,152) .. controls (306.42,152) and (310,155.58) .. (310,160) -- (310,319) .. controls (310,323.42) and (306.42,327) .. (302,327) -- (278,327) .. controls (273.58,327) and (270,323.42) .. (270,319) -- (270,160) .. controls (270,155.58) and (273.58,152) .. (278,152) -- cycle ;
\draw    (206,164.5) .. controls (246,134.5) and (240,213.5) .. (280,183.5) ;
\draw    (208,226.5) .. controls (248,196.5) and (239,235.5) .. (279,205.5) ;
\draw    (205,305.5) .. controls (245,275.5) and (238,334.5) .. (278,304.5) ;
\draw  [dash pattern={on 0.84pt off 2.51pt}]  (297,172.5) -- (381,131.5) ;
\draw  [dash pattern={on 0.84pt off 2.51pt}]  (301,178.5) -- (375,176.5) ;
\draw  [dash pattern={on 0.84pt off 2.51pt}]  (301,200.5) -- (377,198.5) ;
\draw  [dash pattern={on 0.84pt off 2.51pt}]  (303,209.5) -- (376,242.5) ;
\draw  [dash pattern={on 0.84pt off 2.51pt}]  (308,296.5) -- (373,298.5) ;
\draw  [dash pattern={on 0.84pt off 2.51pt}]  (297,307.5) -- (371,343.5) ;
\draw    (346,346.5) -- (375,369) ;
\draw    (352,371.5) -- (374,402) ;

\draw (284,165.4) node [anchor=north west][inner sep=0.75pt]    {$a_{1}$};
\draw (281,192.4) node [anchor=north west][inner sep=0.75pt]    {$a_{2}$};
\draw (274,285.4) node [anchor=north west][inner sep=0.75pt]    {$a_{t_{0} +1}$};
\draw (193,157.4) node [anchor=north west][inner sep=0.75pt]    {$v_{1}$};
\draw (192,220.4) node [anchor=north west][inner sep=0.75pt]    {$v_{2}$};
\draw (184,300.4) node [anchor=north west][inner sep=0.75pt]    {$v_{t_{0} +1}$};
\draw (282,126.4) node [anchor=north west][inner sep=0.75pt]    {$P_{1}$};
\draw (381,82.4) node [anchor=north west][inner sep=0.75pt]    {$P_{2}$};
\draw (382,118.9) node [anchor=north west][inner sep=0.75pt]    {$b_{1}^{1}$};
\draw (377,232.4) node [anchor=north west][inner sep=0.75pt]    {$b_{t_{0} -1}^{2}$};
\draw (380,188.4) node [anchor=north west][inner sep=0.75pt]    {$b_{1}^{2}$};
\draw (378,163.4) node [anchor=north west][inner sep=0.75pt]    {$b_{t_{0} -1}^{1}$};
\draw (383,134.9) node [anchor=north west][inner sep=0.75pt]    {$\vdots $};
\draw (379,201.9) node [anchor=north west][inner sep=0.75pt]    {$\vdots $};
\draw (376,285.4) node [anchor=north west][inner sep=0.75pt]    {$b_{1}^{t_{0} +1}$};
\draw (375,331.4) node [anchor=north west][inner sep=0.75pt]    {$b_{t_{0} -1}^{t_{0} +1}$};
\draw (382,307.4) node [anchor=north west][inner sep=0.75pt]    {$\vdots $};
\draw (380,252.4) node [anchor=north west][inner sep=0.75pt]    {$\vdots $};
\draw (282,239.4) node [anchor=north west][inner sep=0.75pt]    {$\vdots $};
\draw (380,359.4) node [anchor=north west][inner sep=0.75pt]    {$c_{1}$};
\draw (377,396.4) node [anchor=north west][inner sep=0.75pt]    {$c_{t_{0} -1}$};
\draw (383,375.4) node [anchor=north west][inner sep=0.75pt]    {$\vdots $};

\end{tikzpicture}
		\caption{The paths from $v_i$ to $a_i$ are short elementary segments. For each $a_i$ there are at most ${t_0-1}$ many vertices in $P_2$ that are not adjacent to $a_i$. To not overcomplicate the figure we avoided putting edges. The dotted lines mean \emph{not adjacent}, where there is no dotted line, there is an edge.}\label{Figure: figure}
	\end{figure}
	
\end{proof}
By $\ramsey{m}{n}$ we denote the \emph{Ramsey number} for $m$ and $n$, that is,
the smallest integer $N$ such that every graph of size~$N$ contains either
a clique of size~$m$ or an independent set of size~$n$.
We are ready to prove the Theorem \ref{theorem: main result}. Let 
\begin{displaymath}
	\mathscr{R}^k(m,n) := \underbrace{\ramsey{m}{\ramsey{m}{\cdots,\ramsey{m}{n}}}}_\text{k times}.
\end{displaymath}
\begin{proof}


Let us first prove that uniformly almost-wide classes are strongly flip-flat.
Recall that removing a vertex $v$ can be done by performing a flip between $v$ and the neighbourhood of $v$, i.e., the flip $(\{v\}, N(v))$; note that after this flip the vertex $v$ is completely independent from the rest of the graph and is not adjacent to any other vertex. For a  
 uniformly almost-wide class of graphs $\CC$ there exists an integer $k$ such that for any $r\in \NN$ there exists a function $N_r\colon \NN\rightarrow\NN$, so that for any graph $G\in\CC$ and any set of vertices $A\subseteq V(G)$, with $|A|\geq N_r(m)$, we can find $k$ vertices $v_1,\cdots,v_k$ such that there exists a set of vertices $B\subseteq V(G)\setminus\{v_1,\cdots,v_k\}$, where $|B|\geq m$ and $B$ is $r$-independent in $G\setminus \{v_1,\cdots,v_k\}$. Now, instead of removing $\{v_1,\cdots,v_{k}\}$ we iteratively perform the flips $(\{v_1\}, N(v_1)),\cdots ,(\{v_k\},N(v_k))$, notice that the iterative application of the flips means that $N(v_i)$ is the neighbourhood of $v$ after the first $i-1$ flips. With these flips, we completely isolate $v_1,\cdots,v_k$ from the rest of the graph and therefore $B$ will be $r$-independent in $G\oplus \{(\{v_1\}, N(v_1)),\cdots ,(\{v_k\},N(v_k))\}$. Hence, for any $r\in \NN$ the same function $N_r$ and the same fixed integer $k$ witness strong flip-flatness of $\CC$. In addition, uniformly almost-wide classes of graphs are weakly sparse, as for any integer $k$, we cannot find any set of $2$-independent vertices after removing $k$ vertices in $K_{k+1,k+1}$.

Now, for the other direction, let $k$ be the size of the set of flips given by the strong flip-flatness of $\CC$. We prove this by induction on $k$. For $k=0$, the statement follows. 
As the induction hypothesis, we assume that if a class $\CC$ is weakly sparse, with parameter $t_0$, and for every $r\in \NN$ there exists a function $N_r:\NN\rightarrow \NN$ such that for all $G\in \CC$ and $m\in \NN$ and for all $A\subseteq V(G)$ with $|A|\geq N_r(m)$, we can find a $(r,m)$-flippable instance $(A,\FF)$ with $|\FF|<k$, then $\CC$ is uniformly almost-wide and $s'=(k-1)(t_0^2+t_0-2)$ is the constant witnessing the wideness.

Now, suppose that $\CC$ is weakly sparse, with parameter $t_0$ such that $t_0^2+t_0-2\geq 4t_0+2$, and strongly flip-flat with parameter $k$. We define $N_r'(m) = N_r(\mathscr{R}^k(m, t_0^2+t_0+m-2))$ to be the function in the definition of uniform almost-wideness. For any graph $G\in\CC$ and an arbitrary set $A\subseteq V(G)$ with $|A|\geq N_r'(m)$ we can find a $(r,m)$-flippable instance $(A,\FF)$.
	
If $|\FF|\leq k-1$ then we use the induction hypothesis for this particular set $A$ to find a set $S$ such that $(A,S)$ is a $(r,m)$-widenable instance of $G$.
	
	Hence, we can focus on the case where no subset of $A$ of size $m$ is $r$-independent after any $k-1$ flips; however there exists a set $B\subseteq A$ with $|B|\geq \mathscr{R}^k(m, t_0^2+t_0+m-2)$ such that $B$ is $r$-independent in  $G\oplus\FF$.

	Let us examine the iterative effect of $\FF:=\{F_1,\ldots,F_k\}$ on $B$. As $|B|\geq \mathscr{R}^k(m, t_0^2+t_0+m-2)$ at the beginning we face two
        cases, either $B$ contains a set of vertices of size $m$ where all the elements are pairwise at distance strictly greater than $r$, or it contains a set of vertices of size $\mathscr{R}^{k-1}(m, t_0^2+t_0+m-2)$ where all the elements are pairwise at distance less than or equal to $r$. Because of our assumption on the necessity of at least $k$ flips, the former can not happen and so there exists $B_0\subseteq B$ with $|B_0|\geq \mathscr{R}^{k-1}(m, t_0^2+t_0+m-2)$ such that all its elements are pairwise at distance less than or equal to $r$ in $G$. We use the same argument on $B_0$ but in $G\oplus \flip{F_1}$, this gives us a set $B_1\subseteq B_0$ with $|B_1|\geq \mathscr{R}^{k-2}(m, t_0^2+t_0+m-2 )$ such that all its elements are pairwise at distance less than or equal to $r$ in $G\oplus\flip{F_1}$. We repeat similarly and get $B_{k-1}\subseteq B$ with $|B_{k-1}|\geq t_0^2+t_0+m-2$ such that all its elements are pairwise at distance less than or equal to $r$ in $G\oplus\flip{F_1}\oplus\cdots\oplus\flip{F_{k-1}}$. 
	As $B$ is $r$-independent in $G\oplus \FF$,
	 we conclude that after the application of $\flip{F_k}$ on $G\oplus\flip{F_1}\oplus\cdots\oplus\flip{F_{k-1}}$, all the elements of $B_{k-1}$ should be
         at distance strictly greater than $r$. 
%

 By Lemma \ref{lemma: flips do not intersect}, it is safe to assume that the flips $\flip{F_1},\cdots,\flip{F_k}$ do not alter an edge twice. Hence, $\flip{F_k}$ is acting on the original edges of $G$, which excludes $K_{t_0,t_0}$ as subgraph, and not the edges added by $\flip{F_1},\cdots,\flip{F_{k-1}}$, i.e., there is no $K_{t_0,t_0}$ with its bipartition included in the side of the flip. Additionally, $B_{k-1}$ has size larger than $t^2_0+t_0+m-2$ and is
 $r$-independent in $(G\oplus\flip{F_1}\oplus\cdots\oplus\flip{F_{k-1}})\oplus \flip{F_k}$;
   thus, by Lemma \ref{case of k=1}, we can find sets $S$ with $|S|\leq t_0^2+t_0-2$ and $B'$ with $|B'|\geq m$ as promised. As $(G\oplus\flip{F_1}\oplus\cdots\oplus\flip{F_{k-1}})\setminus S$ and $(G\setminus S)\oplus\flip{F_1}\oplus\cdots\oplus\flip{F_{k-1}}$ are equal, we are now able to use the induction hypothesis on $(G\setminus S)\oplus\flip{F_1}\oplus\cdots\oplus\flip{F_{k-1}}$. Thus, we have a set $B''$ with $|B''|\geq m$ and a set $S'$ with $|S'|\leq (k-1)(t_0^2+t_0-2)$ such that $B''$ is $r$-independent in $(G\setminus S)\setminus S'$, that is $B''$ is $r$-independent in $G\setminus (S\cup S')$ where $|S\cup S'|\leq k(t_0^2+t_0-2)$.

	Therefore, for any $(r,m)$-flippable instance $(A,\FF)$ of $G$ with respect to $N_r'$ and $|\FF|\leq k$, we can find a $(r,m)$-widenable instance $(A,S)$ of $G$ with respect to $s'$ and $N_r'$ where $s'\leq k(t_0^2+t_0-2)$. As $s'$ is independent from $r$, we conclude that $\CC$ is uniformly almost-wide.
      \end{proof}

While our argument extends to show that flip-flat graph classes that are weakly sparse are uniformly quasi-wide, this result is equivalent to the known result that monadically stable and weakly sparse classes are nowhere dense \cite{dvovrak2018induced,podewski1978stable, adler2014interpreting}, so we do not elaborate further.
      
As a corollary, we obtain the following characterization of strongly flip-flat classes that have structurally bounded expansion.  We say that a class $\CC$ of graphs has \emph{bounded expansion} if there exists a function $f \colon \NN \to \NN$ such that for every $r \in \NN$ and every depth-$r$ shallow minor $H$ of a graph $G \in \CC$, we have $\frac{\|H\|}{|H|} \leq f(r)$, where $\|H\|$ denotes the number of edges and $|H|$ the number of vertices of $H$. 

Note that for bounded expansion classes, we require the edge density of depth-$r$ shallow minors to be bounded, whereas for nowhere dense classes, we ask that large cliques be excluded as depth-$r$ shallow minors. Accordingly, every class of bounded expansion is nowhere dense.

Let us recall that a structurally bounded expansion class $\CC$ of graphs is a class of graphs for which we can find a bounded expansion class of graphs $\DD$ and a transduction $T_{\delta,\varphi}$ such that $\CC\subseteq T_{\delta,\varphi}(\DD)$. 
\begin{theorem}\label{thm:struct-bd-expansion} A structurally bounded expansion graph class $\mathscr{C}$ is strongly flip-flat if and only if there is a uniformly 
  almost-wide graph class $\mathscr{D}$ of bounded expansion and a transduction $\tau$ such that $\mathscr{C}\subseteq \tau(\mathscr{D})$.
\end{theorem}

\begin{proof} If $\mathscr{C}$ is structurally bounded expansion, then there are transductions $\tau_1$ and $\tau_2$ such that
  $\tau_1(\mathscr{C})$ has bounded expansion and $\mathscr{C}\subseteq \tau_2(\tau_1(\mathscr{C}))$ \cite{gajarsky2020first}. Because transductions preserve
  strong flip-flatness (Proposition \ref{prop:transduce-sflipflat}), we
  have that $\tau_1(\mathscr{C})$ is strongly flip-flat and by Theorem \ref{theorem: main result}, we can conclude that $\tau_1(\mathscr{C})$ is uniformly almost-wide.  
  \end{proof} 
  Theorem \ref{thm:struct-bd-expansion} confirms Conjecture~\ref{conjecture} for graph classes of structurally bounded expansion.

\section{Conclusion}
A corollary of our main theorem is the confirmation of Conjecture \ref{conjecture} on strongly flip-flat graph classes that have structurally bounded expansion, partly
answering the characterisation question raised in \cite{eleftheriadis2024extension}. It seems however that there is
still a long way to completely characterise strongly flip-flat classes of graphs. A confirmation of our conjecture \ref{conjecture} will give a
characterisation of strongly flip-flat graph classes as its reverse has been proved (Proposition \ref{prop:transduce-sflipflat}).

In order to prove Conjecture \ref{conjecture}, we must find a nowhere dense preimage for any class $\CC$ of strongly flip-flat graph classes and then
argue that this preimage itself is strongly flip-flat, \ie, equivalent of saying that it is almost-wide. This might be as hard as the
\emph{sparsification conjecture} \cite{ossonademendez:LIPIcs.STACS.2021.2}, stated below. 
\begin{conjecture}
Every monadically stable graph class is structurally nowhere dense.	
\end{conjecture}

A first step towards Conjecture \ref{conjecture} is probably assuming that the sparsification conjecture is true, and then prove the following, that we
state as a conjecture.

\begin{conjecture}[Weak version of Conjecture \ref{conjecture}]
Let $\CC$ be a strongly flip-flat class of graphs and $\DD$ be a nowhere dense class of graphs. If there exists a transduction $T_{\delta,\varphi}$ such that
$\CC\subseteq T_{\delta,\varphi}(\DD)$, then there exists a class $\DD'$ of uniformly almost-wide graphs and a transduction $T'_{\delta,\varphi}$ such that $\CC\subseteq
T'_{\delta,\varphi}(\DD')$.
\end{conjecture}

Our conjecture may not be the easiest way to characterise strongly flip-flat graph classes. One can, for instance, look into forbidden induced subgraphs.  In the following we give an example of a nowhere dense class of graphs that is not strongly flip-flat. We denote by $G^s$ the $s$-subdivision of
 a graph $G$. 

 \begin{theorem}\label{Example: not SFF} Let $\CC:=\{K^s_{s,N}\mid s\in \NN,~~N\in\NN\}$. Then $\CC$ is not strongly flip-flat.

 \end{theorem}
 \begin{proof}
 Every graph $G$ in $\CC$ excludes $K_{2,2}$ as subgraph; hence $\CC$ is weakly sparse. Note that $K^{2s}_s$ is the largest $2s$-subdivided clique, where the main vertices of the clique are the $s$ vertices of the small side of $K^{s}_{s,N}$. Therefore, $\CC$ is nowhere dense as well.
 
  Let $r:= s+1$. We pick $K^s_{s,N}$ with an arbitrary large $N$ and an arbitrary large set
 of vertices $A$ from the large side of $K_{s,N}^s$. To make any subset of $A$ $r$-independent, we need to remove all the $s$ vertices of the other
 side. Therefore $\CC$ is not almost-wide (one can infer that $\CC$ is not almost-wide by \cite[Theorem 8.4]{nevsetvril2012sparsity}). 	By Theorem \ref{theorem: main result}, $\CC$ is not strongly flip-flat.
 \end{proof}

As an alternative, one could look for a characterisation via forbidden induced subgraphs occurring in flips, as in the case of monadic stability provided in \cite{dreier2024flip} or bounded shrubdepth provided in \cite{DBLP:conf/icalp/Mahlmann25}, or in terms of shallow vertex-minors, as in \cite{buffiere2024shallow}.

Characterisations can be combinatorial and in terms of games, such as the flipper game for flip-flat classes of graphs
\cite{torunczyk2023flip}. 
We remark that the flipper cannot, in general, win in a constant number of rounds---independent of the radius---on strongly flip-flat classes of graphs. This game, in which the flipper has a strategy requiring a constant number of rounds independent of the radius, corresponds to the radius infinity game; as we can consider radius $n$ when playing on an $n$-vertex graph. This characterises instead the class of graphs with
  bounded shrubdepth.


\bibliography{Bibliography.bib}

\end{document}